\documentclass[article]{IEEEtran}
\usepackage{epsfig,latexsym,amsmath,amssymb,setspace,cite,color,graphicx,algorithm,algorithmic,cases, subcaption}

\usepackage{physics}
\newcommand\numberthis{\addtocounter{equation}{1}\tag{\theequation}}
\makeatletter

\newcommand{\pushright}[1]{\ifmeasuring@#1\else\omit\hfill$\displaystyle#1$\fi\ignorespaces}
\makeatother

\usepackage{amsmath}
\usepackage{amsfonts}
\usepackage{amssymb}
\usepackage{dsfont}
\usepackage{bm}
\usepackage{amsthm}
\usepackage{newlfont}
\usepackage{float}
\usepackage{hyperref}
\usepackage{algorithm}
\usepackage{algorithmic}
\usepackage{enumerate}
\usepackage{chngcntr}
\usepackage{mathtools}
\usepackage{breqn}
\usepackage{stmaryrd}
\usepackage{cite}
\hypersetup{
    colorlinks=true,  
    linkcolor= blue,  
    citecolor=blue,  
        urlcolor=black  
}

\DeclareMathOperator*{\argmax}{arg\,max}

\begin{document}
\sloppy
\allowdisplaybreaks[1]

\newtheorem{thm}{Theorem} 
\newtheorem{lem}{Lemma}
\newtheorem{prop}{Proposition}
\newtheorem{cor}{Corollary}
\newtheorem{defn}{Definition}
\newcommand{\remarkend}{\IEEEQEDopen}
\newtheorem{remark}{Remark}
\newtheorem*{rem}{Remark}
\newtheorem*{ex}{Example}
\newtheorem{pro}{Property}

\newenvironment*{example}[1][Example]{\begin{trivlist}
\item[\hskip \labelsep {\bfseries #1}]}{\end{trivlist}}

\renewcommand{\qedsymbol}{ \begin{tiny}$\blacksquare$ \end{tiny} }

\renewcommand{\algorithmicrequire}{\textbf{Input:}}
\renewcommand{\algorithmicensure}{\textbf{Inputs:}}

\renewcommand{\leq}{\leqslant}
\renewcommand{\geq}{\geqslant}

\title {Secure Distributed Storage: Optimal Trade-Off Between Storage Rate and Privacy Leakage 
}

\author{ R\'emi A. Chou, J\"{o}rg Kliewer, \emph{Fellow, IEEE} \thanks{R. Chou is with the Department of Computer Science and Engineering, The University of Texas at Arlington, Arlington, TX. J. Kliewer is with the Department of Electrical and Computer
Engineering, New Jersey Institute of Technology, Newark, NJ.  This work was supported in part by NSF grants CCF-2201824 and CCF-2201825. E-mails: remi.chou@uta.edu, jkliewer@njit.edu. A preliminary version of this work has been presented at the 2020 IEEE International Symposium on
Information Theory (ISIT)~\cite{chou2020secure}.}  }

\maketitle

\begin{abstract}

Consider the problem of storing data in a distributed manner over $T$ servers. Specifically, the data needs to (i) be recoverable from any $\tau$ servers, and (ii) remain private from any $z$ colluding servers, where privacy is quantified in terms of mutual information between the data and all the information available at any $z$ colluding servers. For this model, our main results are (i) the fundamental trade-off between storage size and the level of desired privacy, and (ii) the optimal amount of local randomness necessary at the encoder. As a byproduct, our results provide an optimal lower bound on the individual share size of ramp secret sharing schemes under a more general leakage symmetry condition than the ones previously considered in the literature. 
\end{abstract}

\section{Introduction}

Centralized data storage of sensitive information means compromising the entirety of the data in
the case of a data breach. By contrast, well-known distributed storage strategies, where data are
stored in multiple servers, can offer resilience against data breaches at a subset of servers and avoid having a single point of failure. Secure distributed storage schemes, e.g.,~\cite{garay2000secure,rawat2016centralized,bitar2018staircase,huang2016communication}, often rely on the idea of secret sharing as introduced in~\cite{shamir1979share,blakley1979safeguarding} -- we refer to \cite{beimel2011secret} for a comprehensive literature review on secret sharing. Hence, there is a fundamental lower bound on the required storage space necessary to securely store information in a distributed manner. Specifically, in any threshold secret sharing scheme, the total amount of information that needs to be stored must at least be equal to the entropy of the secret times the number of participants, see e.g., \cite{karnin1983secret}, and it is thus impossible to reduce the storage space without any  changes to the model assumptions.

  In this paper, we propose to determine the optimal cost reduction, in terms of storage space, that can be obtained in exchange of tolerating a \emph{controlled amount} of reduced privacy.  Specifically, we focus on a setting where a file $F$ needs to be stored at $T$ servers. The file must be recoverable from  $\tau$ servers, and needs to remain private from any $z$ colluding servers. Here, privacy is quantified in terms of mutual information between the data and all the information available at any $z$ colluding servers. In particular, we introduce a parameter $\alpha \in [0,1]$, to be chosen by the system designer, and require that no more than a fraction $\alpha$ of the file can be learned by a set of $z$ colluding servers. As a function of the parameters $(\tau,z,\alpha)$, we establish the optimal sum of the share sizes and the optimal amount of local randomness needed at the encoder. Under the assumption of leakage symmetry, i.e., when the information leakage about the file at a given set of colluding servers only depends on the cardinality of the set and not on the identities of the servers among this set, we establish the optimal individual share size for each server. Secret sharing schemes that satisfy such a leakage symmetry are also referred to as uniform secret sharing schemes, e.g.,~\cite{yoshida2018optimal,farras2014optimal}.

\subsection{Previous work}
Secret sharing was first introduced in \cite{blakley1979safeguarding,shamir1979share} and provides perfect security in that any set of colluding participants that is not allowed to reconstruct the secret cannot learn, in an information-theoretic sense, any information about the secret. With the objective to reduce the size of the participants' shares, ramp secret sharing has then been introduced in \cite{yamamoto1986secret,blakley1984security} to relax the security guarantees of secret sharing schemes. Specifically, in \cite{yamamoto1986secret,blakley1984security}, any set of colluding users with size smaller than some parameter $z$ cannot learn any information about the secret, any set of colluding users with size larger than or equal to some parameter $\tau$ can reconstruct the secret, and any set of colluding users with size striclty larger than $z$ but strictly smaller than $\tau$ can learn part of the secret. Additionally, for this last type of set of colluding users, the information leakage about the secret grows linearly with its size. Later, this idea to reduce the size of the shares by allowing information leakage was generalized to non-linear access function, e.g.,~\cite{yoneyama2004non,yoshida2007secure}, and further studied under the term \emph{non-perfect secret sharing}, e.g., \cite{yoshida2018optimal,farras2014optimal,kurosawa1994nonperfect}.

Performance metrics of interest for any secret sharing include the characterization of the optimal size of the shares. More specifically, characterizing the optimal sum of the share sizes, the optimal maximal share size, or the optimal individual share sizes for various secret sharing  settings has been the subject of intense research, e.g., \cite{karnin1983secret,capocelli1993size,blundo2001information,van1995information,csirmaz1995size,blundo1997tight,farras2014optimal,yoshida2018optimal,farras2020improving}, we also refer to \cite{beimel2011secret} for a survey of known results. For instance, optimal individual share sizes have been characterized for the original secret sharing setting \cite{blakley1979safeguarding,shamir1979share}, e.g., \cite{karnin1983secret}, and the optimal sum of the share sizes for ramp secret sharing has also been fully characterized, e.g., \cite{blundo1993efficient}. Unfortunately, even for the relatively simple case of ramp secret sharing the full characterization of optimal individual share sizes is unknown and a challenging problem. A partial solution is proposed in \cite{blundo1993efficient} by restricting the class of ramp secret sharing scheme to the class of linear ramp secret schemes, which adds symmetry to the problem. More recently, a more general and natural definition of symmetry is introduced in the form of uniform secret sharing~\cite{yoshida2012optimum,farras2014optimal,yoshida2018optimal}. In these works, the authors characterize the optimal individual share sizes obtained under such a symmetry condition, which requires that the information leakage only depends on the size of the set of colluding participants, and not on the specific identities of the participants in this set. 

\subsection{Comparison to previous work}
  
  Our setting not only considers generalizations of previous secret sharing settings but also introduces a fundamentally different view to study the trade-off between storage needs and privacy leakage. Specifically,  
   a major difference between the study of uniform secret sharing in  \cite{yoshida2018optimal,farras2014optimal} and our work is that, in \cite{yoshida2018optimal,farras2014optimal}, optimal individual share sizes are derived for a fixed access function, i.e., the information leakage about the file tolerated at a given set of colluding servers is a fixed and given value.  In contrast, in our setting we derive optimal individual share sizes for secret
sharing schemes whose access functions are not fixed but are allowed to belong to a set of access
functions. Indeed, in our setting, only two points of the access functions are fixed as parameters:
one point indicates a reconstruction threshold $\tau$, and the other  point indicates a maximum number of
colluding servers $z$. All the other points of the access function are optimized to minimize the
share sizes. This difference introduces a non-trivial optimization problem over a set of access functions
to determine optimal individual share sizes. We show that this optimization reduces to  maximizing the sum of consecutive gradients of an access function over the set of all possible access functions that satisfy our problem constraints, i.e., it must be less than or equal to $\alpha$ in point $z$ and equal to one in point $\tau$. The crux to solve this optimization is to introduce the concave envelopes of the access functions to show that piecewise linear access functions are solutions to the optimization.

We note that the idea of trading storage space against information leakage is  also closely related to non-perfect secret sharing \cite{farras2016recent,yoshida2018optimal,farras2014optimal}, including ramp secret sharing with linear \cite{yamamoto1986secret,blakley1984security} or non-linear access functions~\cite{yoneyama2004non,yoshida2007secure}. Similar to our previous comment, these settings have been studied for fixed access functions, whereas, in this study, to minimize share sizes, we consider secret sharing schemes with access functions allowed to belong to a set of access functions.

While the above remark on non-perfect secret sharing applies ramp secret sharing schemes, e.g., \cite{yamamoto1986secret,blakley1984security}, we highlight a new result for ramp secret sharing that follows from our main results. Specifically, when the privacy parameter is $\alpha =0$, i.e., perfect privacy is required against $z$ colluding servers, our results prove that
among all uniform ramp secret sharing schemes, which represents a more general class of secret sharing schemes than linear ramp secret sharing schemes, the ones that have a piecewise linear access function have
the minimum individual share sizes. This result had also not been previously proved in the literature.

\subsection{Paper organization}  

We formulate our problem statement and review known results in Section \ref{secs}. We present our main results in Section \ref{secres} and relegate the proofs to Sections~\ref{secproofach}, \ref{secproofc1}  to streamline presentation. Finally, we provide concluding remarks in Section~\ref{concl}.

\section{Problem statement and review of known results} \label{secs}
Notation: Let $\mathbb{N}$, $\mathbb{R}$, and $\mathbb{Q}$ be the sets of natural, real, and rational numbers, respectively. For $a,b \in \mathbb{R}$, define $\llbracket a,b \rrbracket \triangleq [\lfloor a \rfloor , \lceil b \rceil ] \cap \mathbb{N}$ and $[a]^+\triangleq\max(0,a)$. For two arbitrary sets $\mathcal{S}$ and $\mathcal{T}$, a sequence of elements $x_t\in\mathcal{S}$, $t\in\mathcal{T}$,  indexed by the set $\mathcal{T}$ is written as $(x_t)_{t\in\mathcal{T}}$. 

\subsection{Problem statement} \label{sec:ps}
Consider $T \geq 2$ servers indexed by $\mathcal{T} \triangleq \llbracket 1,T\rrbracket$. For $t\in\mathcal{T}$, define $[\mathcal{T}]^{\geq t}$ as the set of all the subsets of $\mathcal{T}$ that have a cardinality larger than or equal to $t$, i.e., $[\mathcal{T}]^{\geq t}\triangleq \{  \mathcal{S} \subseteq \mathcal{T} : |\mathcal{S}| \geq t\}$. Similarly, define $[\mathcal{T}]^{\leq t}\triangleq \{  \mathcal{S} \subseteq \mathcal{T} : |\mathcal{S}|\leq t\}$ and $[\mathcal{T}]^{= t}\triangleq \{  \mathcal{S} \subseteq \mathcal{T} : |\mathcal{S}| = t\}$. 
\begin{defn}
Let $(\lambda_t)_{t\in \mathcal{T}}\in \mathbb{N}^T$, $\rho \in \mathbb{N}$, and $\tau \in \mathcal{T}$. A $(\tau,(\lambda_t)_{t\in \mathcal{T}},\rho)$ coding scheme consists of 
\begin{itemize}
\item A file $F$, which is represented by a random binary sequence with finite length;
\item Local randomness in the form of a sequence $R$ of $\rho$ bits uniformly distributed over $\{ 0,1\}^{\rho}$ and independent of $F$;
		\item $T$ encoders $(e_t)_{t\in\mathcal{T}}$, where for $t\in \mathcal{T}$,		
		\begin{align*}
		e_t	 : \{ 0,1\}^{|F|} \times \{ 0,1\}^{\rho} \to \{ 0,1\}^{\lambda_t}, 
		 (F,R)  \mapsto  M_t,
	\end{align*}	 
	  which takes as input the file $F$ and the local randomness $R$, and outputs the sequence $M_t$, referred to as share in the following, of length $\lambda_t \in \mathbb{N}$. $\lambda_t$ is referred to as share size in the following.  
	  \item $T$ servers, where Server $t \in \mathcal{T}$ stores $M_t$. In the following, for any subset $\mathcal{S} \subseteq \mathcal{T}$ of servers, we use the notation $M_{\mathcal{S}} \triangleq (M_t)_{t \in \mathcal{S}}$.
	  \item For any subset $\mathcal{S} \subseteq \mathcal{T}$ such that $|\mathcal{S}|\geq \tau$, a decoder \begin{align*}
\smash{	d_{\mathcal{S}}	 : \bigtimes_{t\in \mathcal{S}} \{ 0,1\}^{\lambda_t}  \to \{ 0,1\}^{|F|},M_{\mathcal{S}} \mapsto  \hat{F}(\mathcal{S}),}
	\end{align*}	 
	  which takes as input $M_{\mathcal{S}}$ and  outputs $\hat{F}(\mathcal{S})$, an estimate of $F$.
\end{itemize}
\end{defn}

\begin{defn} \label{def}
For $\tau \in \mathcal{T}$, $\alpha \in \mathbb{Q}\cap[0,1]$, and $z \in \llbracket 1 , \tau -1 \rrbracket$, a $(\tau,(\lambda_t)_{t\in\mathcal{T}},\rho)$ coding scheme is $(\alpha,z)$-private if
\begin{align}
\displaystyle\max_{\mathcal{S} \in [\mathcal{T}]^{\geq\tau}}	H(F|\hat{F}(\mathcal{S})) & = 0, \text{ (Recoverability)} \label{eqreq1} \\
	 \displaystyle\max_{\mathcal{S} \in [\mathcal{T}]^{\leq z}}\frac{I(F;M_{\mathcal{S}})}{H(F)} &\leq \alpha, \text{ (Privacy)} \label{eqreq2} .
\end{align}
\end{defn}
Requirement \eqref{eqreq1} means that any subset of $\tau$ or more servers can reconstruct the file $F$. Note that Requirement \eqref{eqreq1} implies $\displaystyle\max_{\mathcal{S} \in [\mathcal{T}]^{\geq\tau}}	H(F|M_{\mathcal{S}}) = 0$. Requirement~\eqref{eqreq2} means that any subset of servers with size smaller than or equal to $z$ must not learn more than $\alpha H(F)$ bits of information about $F$. In the following, $\tau$ is referred to as reconstruction threshold, $\alpha$ is referred to as privacy leakage parameter, and $z$ is referred to as privacy threshold. The setting is illustrated in Figure~\ref{figsetting} when $(T,\tau,z) = (3,3,2)$.

\begin{figure}
    \centering
    \begin{subfigure}[t]{0.45 \textwidth}
        \centering
    \includegraphics[width=4.2cm]{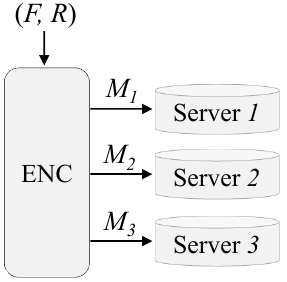}
        \caption{Storage}
        \vspace{1.9em}
    \end{subfigure}
    \begin{subfigure}[t]{0.45 \textwidth}
        \centering
    \includegraphics[width=6cm]{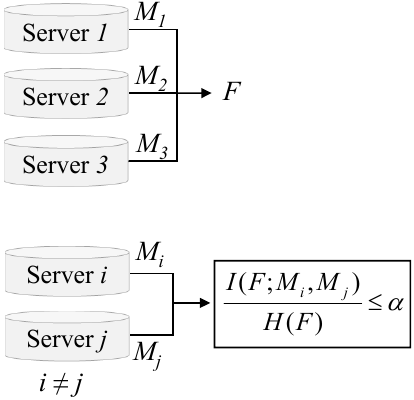}
        \caption{Retrieval}
    \end{subfigure}
\caption{Secure distributed storage (a) and retrieval (b) with privacy leakage for $T=3$ servers, reconstruction threshold $\tau =3$,  privacy threshold $z=2$, and privacy leakage parameter $\alpha$. $M_i$ is stored at Server $i\in \{1,2,3\}$ and created from the File $F$ and the local randomness $R$ available at the encoder.}\label{figsetting}
\end{figure}

\begin{rem}
In Definition \ref{def}, $\alpha$ is restricted to be a rational number. However, note that by density of $\mathbb{Q}$ in $\mathbb{R}$, for any $\beta \in [0,1]$, for any $\epsilon>0$, there exists $\alpha \in \mathbb{Q}\cap[0,1]$ such that $|\alpha - \beta| \leq \epsilon$.	
\end{rem}

\begin{defn} \label{def3}
Let $\tau \in \mathcal{T}$, $\alpha \in \mathbb{Q}\cap[0,1]$, and $z \in \llbracket 1 , \tau -1 \rrbracket$. Then, for $t\in\mathcal{T}$, define
\begin{align*}
&\lambda_t^{\star} (\alpha,z,\tau) \\
&\triangleq \min \{ \lambda_t \in \mathbb{N} : \text{there exists an $(\alpha,z)$-private}\\
&\phantom{-------}{\text{$(\tau,(\lambda_{t'})_{t'\in\mathcal{T}},\rho)$ coding scheme}}\\
&\pushright{\text{for some $\rho \in\!\mathbb{N}$ and $(\lambda_{t'})_{t'\in\mathcal{T}\backslash \{t\}}\in\!\mathbb{N}^{T-1}$} \} ,}\\
&\lambda_{\textup{sum}}^{\star} (\alpha,z,\tau)\\
&\triangleq \min \{ \textstyle\sum_{t\in\mathcal{T}}\lambda_t \in \mathbb{N} : \text{there exists an $(\alpha,z)$-private} \\
&\pushright{\text{$(\tau,(\lambda_t)_{t\in\mathcal{T}},\rho)$ coding scheme for some $\rho \in\!\mathbb{N}$} \} ,}\\
&\rho^{\star} (\alpha,z,\tau)  \\
&\triangleq \min \{ \rho \in \mathbb{N} : \text{there exists an $(\alpha,z)$-private} \\
&\pushright{\text{$(\tau,(\lambda_t)_{t\in\mathcal{T}},\rho)$ coding scheme for some $(\lambda_t)_{t\in\mathcal{T}}\in\!\mathbb{N}^T$} \}.}
\end{align*}
\end{defn}
For fixed $T$, $\alpha$, $\tau$, and $z$ as in Definition \ref{def3}, our objective in this paper is to characterize the optimal storage size $\lambda_t^{\star} (\alpha,z,\tau)$ at Server $t\in\mathcal{T}$, the optimal total storage size $\lambda_{\textup{sum}}^{\star} (\alpha,z,\tau)$, and the optimal amount of  local randomness needed at the encoder $\rho^{\star} (\alpha,z,\tau)$. 
 Note that it is a priori unclear whether there exists a coding scheme that can simultaneously achieve  $\lambda_t^{\star} (\alpha,z,\tau)$, $t\in\mathcal{T}$,  $\lambda_{\textup{sum}}^{\star} (\alpha,z,\tau)$, and $\rho^{\star} (\alpha,z,\tau)$. However, our results will prove that such a coding scheme exists.

\subsection{Previous results} \label{secrev}
The special case $\alpha = 0$ has been studied in the literature and corresponds to ramp secret sharing \cite{yamamoto1986secret,blakley1984security}. Specifically,  by choosing $\alpha =0$ and $z = \tau -L$, for some $L \in \llbracket 1, \tau -1 \rrbracket$, the problem statement of Section \ref{sec:ps} describes a so-called $(\tau , L, T)$ ramp secret sharing scheme. Additionally, for ramp secret sharing, we have, e.g., \cite{blundo1996randomness,blundo1993efficient},
\begin{align*}
\lambda_{\textup{sum}}^{\star} (\alpha=0,z= \tau -L,\tau) &=  \frac{T}{L} H(F),\\
\rho^{\star} (\alpha=0,z= \tau -L,\tau) &= \frac{\tau-L}{L} H(F).
\end{align*}
As remarked in \cite{blundo1993efficient}, in general, one does not have $\lambda_t^{\star} (\alpha=0,z= \tau -L,\tau) = \frac{1}{L}H(F), \forall t \in \mathcal{T}$, as for some $t\in \mathcal{T}$, the share size could be zero. For this reason,  \cite{blundo1993efficient} considers linear ramp secret sharing schemes, where the leakage on the  file $F$ for a set $\mathcal{S}$ of colluding servers scales linearly
with the size of $\mathcal{S}$ between $\tau -L$ to~$\tau$. In other words, a linear ramp secret sharing satisfies the condition 
\begin{align} \label{eqlins}
\forall \mathcal{S} \in [\mathcal{T}]^{\geq \tau-L+1} \cap [\mathcal{T}]^{\leq \tau-1}, H(F |M_{\mathcal{S}}) = \frac{\tau-|\mathcal{S}|}{L} H(F). \tag{\text{$A_1$}}
\end{align}
For such linear ramp secret sharing schemes,  \cite[Th. 3.3]{blundo1993efficient} establishes the following optimal individual share size:
\begin{align*}
\lambda_t^{\star} (\alpha=0,z= \tau -L,\tau) &= \frac{1}{L} H(F), \forall t \in \mathcal{T} .
\end{align*}

Remark that the definition of linear secret sharing schemes means that a fixed value is assigned to the information leakage at a given set of colluding servers, i.e., \eqref{eqlins} can be rewritten as $$\forall \mathcal{S} \in [\mathcal{T}]^{\geq \tau-L+1} \cap [\mathcal{T}]^{\leq \tau-1}, I(F ;M_{\mathcal{S}}) =  \frac{|\mathcal{S}|-(\tau-L)}{L} H(F). $$

\subsection{Discussion of leakage symmetry conditions used in previous work}
For any $\tau \in \mathcal{T}$, $\alpha \in \mathbb{Q}\cap[0,1]$, and $z \in \llbracket 1 , \tau -1 \rrbracket$, we will establish the optimal individual share size $\lambda_t^{\star} (\alpha,z,\tau)$ for any $t\in\mathcal{T}$ under  the following leakage symmetry condition \eqref{eqlt}
\begin{align} \label{eqlt}
\forall	t \in  \mathcal{T},\exists C_{t} \in \mathbb{R}^+ , \forall \mathcal{S} \in [\mathcal{T}]^{=t},\frac{I(F;M_{\mathcal{S}})}{H(F)} &= C_{t}, \tag{\text{$A_2$}}
\end{align}
where, by convention, we define $C_0 \triangleq 0$.   Condition \eqref{eqlt} means that when considering a subset of servers $\mathcal{S} \subseteq \mathcal{T}$, the privacy leakage about $F$, i.e., $I(F;M_{\mathcal{S}})$, must only depend on the cardinality of $\mathcal{S}$ and not the specific members in $\mathcal{S}$. Note that after normalization by $H(F)$, $\frac{I(F;M_{\mathcal{S}})}{H(F)} \in [0,1]$ for any $\mathcal{S} \subseteq \mathcal{T}$. Note also that, by \eqref{eqreq2},  we must have $C_{t} \leq \alpha$  for any $t \in \llbracket 1 , z \rrbracket$.

In the special case $\alpha =0$, observe that Condition~\eqref{eqlt} is more general than Condition~\eqref{eqlins}, which is reviewed in Section~\ref{secrev} and used to derive the optimal size of individual share for linear secret sharing schemes. Indeed, Condition \eqref{eqlins} is recovered by setting $C_t \triangleq \frac{t - (\tau -L)}{L} $ for $t\in\llbracket \tau-L+1, \tau-1 \rrbracket$ in Condition~\eqref{eqlt} with $L \triangleq \tau - z$. Hence, when $\alpha=0$, Condition~\eqref{eqlt} describes a class of ramp secret sharing schemes that contains linear ramp secret sharing schemes.

Note that the leakage symmetry condition \eqref{eqlt} is introduced under the term \emph{uniform secret sharing} in~\cite{yoshida2018optimal}, where the adjective uniform is used in \cite{yoshida2018optimal} to reflect that \eqref{eqlt} holds. In~\cite{yoshida2018optimal}, the optimal share size is established when  the constants $(C_t)_{t \in \mathcal{T}}$ in \eqref{eqlt}
are fixed. By contrast, in this paper, we are interested in finding the constants $(C_t)_{t \in \mathcal{T}}$ that minimize the individual share size and the necessary amount of local randomness at the encoder. To this end, we will carry an optimization over all possible secret sharing schemes that satisfy the leakage symmetry condition~\eqref{eqlt}.  
Another difference between our study and  \cite{yoshida2018optimal} is that our study extends~\cite{yoshida2018optimal} in the following aspects: 
for $\alpha \neq 0$, we study the optimal sum of the share sizes at all the servers and the optimal amount of local randomness required at the encoder in the absence of any leakage symmetry condition.

\section{Main results} \label{secres}

 We first establish in Theorem \ref{thconverse} the optimal individual share size and optimal amount of local randomness under the leakage symmetry condition \eqref{eqlt}. We then derive three corollaries from Theorem~\ref{thconverse} that recover or extend known results, as outlined below.

\begin{thm} \label{thconverse}
Let $\tau \in \mathcal{T}$, $\alpha \in \mathbb{Q}\cap[0,1]$, and $z \in \llbracket 1 , \tau -1 \rrbracket$. Suppose that the leakage symmetry condition~\eqref{eqlt} holds. Then, for any $t\in\mathcal{T}$, we have
\begin{align*}
\frac{\lambda_t^{\star} (\alpha,z,\tau)}{H(F)} &= \max \left( \frac{1-\alpha}{\tau-z}, \frac{1}{\tau} \right) = \begin{cases} \frac{1-\alpha}{\tau-z} & \text{if }\alpha < \frac{z}{\tau} \\ \frac{1}{\tau} & \text{if }\alpha \geq \frac{z}{\tau}\end{cases} ,\\
\frac{\rho^{\star} (\alpha,z,\tau)}{H(F)} &=     \frac{[z- \tau \alpha]^+ }{\tau-z} = \begin{cases} \frac{z- \tau \alpha }{\tau-z} & \text{if }\alpha < \frac{z}{\tau} \\ 0 & \text{if }\alpha \geq \frac{z}{\tau}\end{cases} .
\end{align*}	
Moreover, there exists an $(\alpha,z)$-private $(\tau,(\lambda_t^{\star} (\alpha,z,\tau))_{t\in\mathcal{T}},\rho^{\star} (\alpha,z,\tau))$ coding scheme, i.e., $\lambda_t^{\star} (\alpha,z,\tau)$, $t\in\mathcal{T}$, and $\rho^{\star} (\alpha,z,\tau)$ can simultaneously be achieved by a single coding scheme.
\end{thm}

\begin{proof}
The achievability proof of Theorem \ref{thconverse} is detailed in Section \ref{secproofach}. The converse proof of Theorem~\ref{thconverse} is  presented in Section \ref{secproofc1}.  
\end{proof}

As expected, Theorem \ref{thconverse} shows that allowing information leakage, controlled by the parameter $\alpha$, enables a reduction of the individual share size and amount of local randomness needed at the encoder. Theorem \ref{thconverse} also shows the existence of a threshold with respect to $\alpha$. Specifically, when $ \alpha \geq z/ \tau$, then a $(\tau,\tau,T)$ ramp secret sharing is sufficient to achieve the optimal share size $1/ \tau$, since, in that case, the share size of any $z$ colluding users is $z / \tau$ and the privacy condition is immediately satisfied.

\begin{cor}\label{cor1}
 Assume that the privacy leakage is $\alpha =0$ and the privacy threshold is $z = \tau -1$. Observe from \eqref{eqreq1} and \eqref{eqreq2} that, in this case, Condition \eqref{eqlt} is always satisfied, in particular, $C_t =0$ when $t \in \llbracket 1 , \tau -1 \rrbracket$, and $C_t =1$ when $t \in \llbracket \tau, T\rrbracket$.   Then, by Theorem \ref{thconverse}, 
we recover the well-known fact, e.g., \cite[Th. 1]{karnin1983secret}, that the optimal share size is the entropy of $F$ for perfect threshold secret sharing, first introduced in~\cite{shamir1979share,blakley1979safeguarding}. 
\end{cor}

\begin{cor}
Suppose that the leakage symmetry condition~\eqref{eqlt} holds. Assume that the privacy leakage is $\alpha =0$ and the privacy threshold is $z = \tau -L$, for some $L \in \llbracket 1, \tau -1 \rrbracket$.   Then, Theorem~\ref{thconverse} 
recovers  the result in\cite{blundo1993efficient,blundo1996randomness}, for $(\tau , L, T)$ \emph{linear} ramp secret sharing schemes, i.e., secret sharing schemes that satisfy Condition \eqref{eqlins}, and generalizes it to the larger class of uniform secret sharing schemes, i.e., secret sharing schemes that satisfy Condition~\eqref{eqlt}. The result can also be interpreted as follows: Among all uniform secret sharing schemes, linear secret sharing schemes  are optimal in terms of individual share size and local randomness necessary at the encoder.
\end{cor}

\begin{cor}
Assume that the reconstruction threshold is $\tau = T$,  the privacy threshold is $z =T-1$, and the Condition \eqref{eqlt} holds. Then, Theorem~\ref{thconverse} 
recovers the results found in \cite{chou2020secure} and generalizes them to the case where the shares are not assumed to be of equal size in the problem statement.
\end{cor}

We numerically illustrate Theorem \ref{thconverse} in the following example. 
\begin{ex}
For the case of $T=12$ servers and a reconstruction threshold $\tau = 7$, we depict in Figures~\ref{figlambda} and \ref{figrho}, $\frac{\lambda_t^{\star} (\alpha,z,\tau)}{H(F)}$, $t\in\mathcal{T}$, and $\frac{\rho^{\star} (\alpha,z,\tau)}{H(F)}$ obtained in Theorem \ref{thconverse}, respectively, as functions of the privacy leakage parameter $\alpha$ and the privacy threshold $z$.
\end{ex}

\begin{figure}[!tbp]
  \centering
  \begin{minipage}[b]{0.49\textwidth}
    \includegraphics[width=8cm]{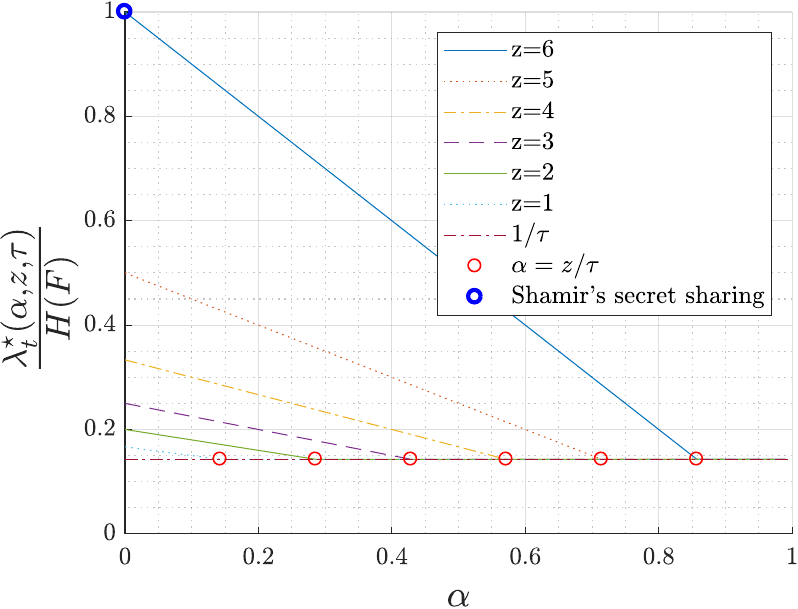}
\caption{$\frac{\lambda_t^{\star} (\alpha,z,\tau)}{H(F)}$, $t\in\mathcal{T}$, when $T=12$, $\tau = 7$, and the privacy threshold belongs to $\llbracket 1, 6 \rrbracket$. The bold blue circle corresponds to the optimal share size for Shamir's secret sharing, as reviewed in Corollary \ref{cor1}.}\label{figlambda}
  \end{minipage}
  \hfill
  \begin{minipage}[b]{0.49\textwidth}
    \includegraphics[width=8cm]{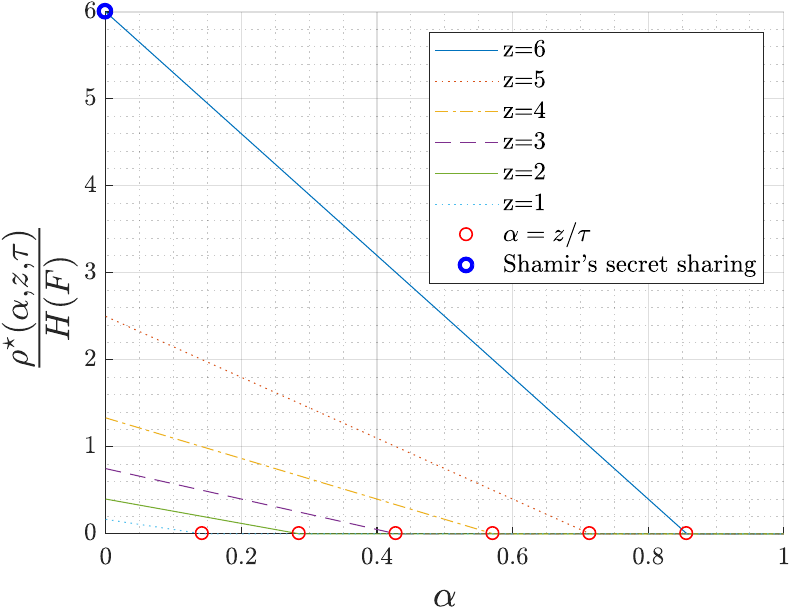}
\caption{$\frac{\rho^{\star} (\alpha,z,\tau)}{H(F)}$ when $T=12$, $\tau = 7$, and $z$ belongs to $\llbracket 1, 6 \rrbracket$. The bold blue circle corresponds to the optimal amount of necessary randomness at the encoder for Shamir's secret sharing, as reviewed in Corollary \ref{cor1}.
}\label{figrho}
  \end{minipage}
\end{figure}

Note that, in the absence of any leakage symmetry condition, the minimum size of an individual share could be zero. For instance, 
using the notation of Section \ref{secrev}, one can construct  a  $(\tau,\tau - z,T)$ ramp secret sharing scheme with $T$ participants, where the share size of $\tau-z-1$ participants is zero  as follows: Consider a $(z+1,1,T-(\tau-z-1))$ ramp secret sharing scheme  with $T-(\tau-z-1)$ participants and consider $\tau-z-1$ additional participants to whom we do not give any share. This shows that Theorem \ref{thconverse} does not hold in the absence of Condition \eqref{eqlt}.

Then, beyond the optimal individual share sizes, we also study the optimal sum of the share sizes in the absence of any leakage symmetry condition. Specifically, in this case, we  establish the optimal sum of the share sizes $ \lambda_{\textup{sum}}^{\star} (\alpha,z,\tau)$ and the optimal amount of local randomness $\rho^{\star} (\alpha,z,\tau)$, both defined in Definition \ref{def3}. 

\begin{thm}  \label{th3}
Let $\tau \in \mathcal{T}$, $\alpha \in \mathbb{Q}\cap[0,1]$, and $z \in \llbracket 1 , \tau -1 \rrbracket$. We have
\begin{align*}
\frac{\lambda_{\textup{sum}}^{\star} (\alpha,z,\tau)}{H(F)} &= T \max \left( \frac{1-\alpha}{\tau-z}, \frac{1}{\tau} \right) = \begin{cases} T\frac{1-\alpha}{\tau-z} & \text{if }\alpha < \frac{z}{\tau} \\ \frac{T}{\tau} & \text{if }\alpha \geq \frac{z}{\tau}\end{cases} ,\\
\frac{\rho^{\star} (\alpha,z,\tau)}{H(F)} &=     \frac{[z- \tau \alpha]^+ }{\tau-z} = \begin{cases} \frac{z- \tau \alpha }{\tau-z} & \text{if }\alpha < \frac{z}{\tau} \\ 0 & \text{if }\alpha \geq \frac{z}{\tau}\end{cases} .
\end{align*}	
Moreover, there exists an $(\alpha,z)$-private $(\tau,(\lambda_t )_{t\in\mathcal{T}},\rho^{\star} (\alpha,z,\tau))$ coding scheme such that $ \sum_{t\in \mathcal{T}}\lambda_t=  \lambda_{\textup{sum}}^{\star} (\alpha,z,\tau)$, i.e., $\lambda_{\textup{sum}}^{\star}(\alpha,z,\tau)$ and $\rho^{\star} (\alpha,z,\tau)$ can simultaneously be achieved by a single coding scheme.
\end{thm}

\begin{proof}
The achievability part follows from the achievability part of Theorem \ref{thconverse} by summing the sizes of the $T$ shares and the converse part is proved in Appendix~\ref{secconvth3}. 
\end{proof}

Note that the converse proof of Theorem \ref{th3} is of combinatorial nature and different from the proof of Theorem~\ref{thconverse}, which involves an optimization problem. Note also that Theorems~\ref{thconverse} and \ref{th3} indicate that there is no gain, in terms of necessary local randomness at the encoder,  between an optimization over secret sharing schemes that satisfy Condition \eqref{eqlt} and any secret sharing schemes that do not.

\section{Achievability proof of Theorem \ref{thconverse}} \label{secproofach}

Let $\alpha \in \mathbb{Q}\cap [0,1]$. We consider the two cases $\alpha \geq z/\tau$ and $\alpha < z/\tau$ in Sections \ref{sec2ach} and \ref{sec1ach}, respectively. We first review some definitions in Section \ref{sec0ach}.  

\subsection{Preliminary definitions} \label{sec0ach}
\subsubsection{Access function} For a coding scheme as in Definition~\ref{def} that satisfies the leakage symmetry Condition \eqref{eqlt}, one can define an access function, e.g.,~\cite{farras2014optimal}, which fully describes the leakage of any set of shares. More specifically, with the notation of Definition \ref{def}, define the access function of a coding scheme as
 \begin{align*}
 g : \llbracket 0 ,T \rrbracket  \to [0,1],  t  \mapsto C_t.
 \end{align*}
Note that, by \eqref{eqlt}, the leakage of any set $\mathcal{S}$ of shares only depends on the cardinality of $\mathcal{S}$. Hence, it is sufficient to consider $g$ to fully describe the leakage of any set of shares. Note also that, by \eqref{eqreq1}, the reconstruction threshold $\tau$ implies that $ g(t) = 1$ for any $t \in \llbracket \tau , T \rrbracket$, and,  by \eqref{eqreq2}, the privacy threshold $z$ implies that $ g(z) \leq \alpha$. Finally, note that by definition of $C_t$, $t\in\mathcal{T}$, in \eqref{eqlt}, $g$ is non-decreasing.

\subsubsection{Ramp secret sharing}
Consider $\tau,T,L \in \mathbb{N}$ such that $ 1\leq L < \tau \leq T$. A $(\tau,L,T)$ linear ramp secret sharing scheme, e.g., \cite{yamamoto1986secret,blakley1984security}, is a coding scheme as in Definition~\ref{def} with $\alpha =0$, $z = \tau - L$, and with access~function 
\begin{align*}
 g : \llbracket 0 ,T \rrbracket  \to [0,1],  
                  t  \mapsto \begin{cases}
0  & \text{if } t \in  \llbracket 0, \tau - L \rrbracket   \\
 \frac{t - \tau +L}{L}    & \text{if } t \in  \llbracket \tau - L +1, \tau  \rrbracket   \\
  1 & \text{if } t \in  \llbracket \tau +1 , T \rrbracket
 \end{cases}.
 \end{align*}
 In particular, any $\tau$ shares can reconstruct $F$, any set of shares less than or equal to $\tau -L$ does not leak any information about $F$, and for sets of shares with cardinality in $ \llbracket \tau - L +1, \tau  \rrbracket$, the leakage increases linearly with the set cardinality.

\subsection{Achievability proof} 
The first step of the achievability scheme is to characterize the access function $g$ of our desired secret sharing scheme. Note that beyond being an access function as defined in Section~\ref{sec0ach}, the only constraint that our setting imposes on $g$ is $g(z) \leq \alpha$. From our converse results, we know that a piecewise linear access function $g$ would provide the lowest possible individual share sizes. We then consider two cases: $\alpha \geq z/\tau$ and $\alpha < z/\tau$. While we handle the first case with a simple ramp secret sharing scheme, we 
 follow the idea from~\cite{yoshida2018optimal} to handle the second case. Specifically, we remark that $g$ can be written as the sum of two other access functions $g_1$ and $g_2$, i.e., $g = g_1 + g_2$, that correspond to two ramp secret sharing schemes. Finally, we construct a secret sharing scheme with access function $g$ by a combination of two ramp secret sharing schemes with access functions $g_1$ and $g_2$.
\subsubsection{Case 1} \label{sec2ach}
Assume that $\alpha \geq z/\tau$.
Let $g$ be the access function of a $(\tau, \tau,T)$ ramp secret sharing scheme, which can be done as in \cite{yamamoto1986secret,blakley1984security} with, for any $t\in \mathcal{T}$,
$\frac{\lambda_t (\alpha,z,\tau)}{H(F)} = \frac{1}{\tau},$ and 
$\frac{\rho (\alpha,z,\tau)}{H(F)} = 0.$

Note that 
this scheme satisfies~\eqref{eqreq1} because $g(t)=1$ for $t \in  \llbracket \tau  , T \rrbracket$, and also satisfies~\eqref{eqreq2} because for any $t \in  \llbracket 0, z \rrbracket$, $g(t) \leq z/\tau\leq  \alpha$. 
\subsubsection{Case 2}\label{sec1ach} Assume that $\alpha < z/\tau$.
Consider the following access function 
\begin{align*}
g: 
 t \mapsto \begin{cases}
\frac{\alpha}{z}t  & \text{if } t \in  \llbracket 0, z \rrbracket   \\
 \frac{1-\alpha}{\tau-z}(t -z) + \alpha   & \text{if } t \in  \llbracket z+1, \tau  \rrbracket   \\
  1 & \text{if } t \in  \llbracket \tau +1 , T \rrbracket
 \end{cases}.
\end{align*}
Note that if one can construct a coding scheme with access function $g$, then this coding scheme satisfies~\eqref{eqreq1} because $g(t)=1$ for $t \in  \llbracket \tau  , T \rrbracket$, and also satisfies \eqref{eqreq2} because for any $t \in  \llbracket 0, z \rrbracket$, $g(t) \leq \alpha$. We construct such a coding scheme using the method in \cite{yoshida2018optimal}. First, note that $g = g_1 + g_2$, where we have~defined
\begin{align*}
g_1:   
 t &\mapsto \begin{cases}
\frac{\alpha}{z}t  & \text{if } t \in  \llbracket 0, \tau \rrbracket   \\
 \frac{\alpha}{z} \tau    & \text{if } t \in  \llbracket \tau+1, T  \rrbracket   
 \end{cases},\\
g_2:  
 t &\mapsto \begin{cases}
 0  & \text{if } t \in  \llbracket 0, z \rrbracket   \\
\frac{1-\alpha}{\tau-z}(t -z) + \alpha  - \frac{\alpha}{z}t & \text{if } t \in  \llbracket z+1, \tau \rrbracket   \\
 1 - \frac{\alpha}{z} \tau    & \text{if } t \in  \llbracket \tau+1, T  \rrbracket   
 \end{cases}.
\end{align*}

Next, we construct a coding scheme with access function $g$ from two ramp secret sharing schemes with the normalized access functions $\left(\frac{\alpha}{z} \tau \right)^{-1}  g_1$ and $\left(1-\frac{\alpha}{z} \tau  \right)^{-1}g_2$.
By \cite{yamamoto1986secret,blakley1984security},  there exist a prime $q$ and $n' \in \mathbb{N}$ such that one can construct an optimal $(\tau,\tau,T)$ ramp secret sharing (with access function $\left(\frac{\alpha}{z} \tau \right)^{-1}  g_1$) that uses $\rho^{(1)}$ random symbols at the encoder and yields the shares $(M^{(1)}_t)_{t\in\mathcal{T}}$ for a secret $F_1 \in \textup{GF}(q^{n_1})$ with $n_1 = \frac{\alpha}{z}\tau  n'$ and $\rho^{(1)}=0$, and a $(\tau,\tau -z,T)$ ramp secret sharing (with access function $\left(1-\frac{\alpha}{z} \tau  \right)^{-1}g_2$) that uses $\rho^{(2)}$ random symbols at the encoder and yields the shares $(M^{(2)}_t)_{t\in\mathcal{T}}$ for a secret $F_2 \in \textup{GF}(q^{n_2})$, independent of $F_1$, with $n_2 =\left(1-\frac{\alpha}{z} \tau  \right) n'$ and $\rho^{(2)}=H(F_2) \frac{z}{\tau -z}$. Then, define $F \triangleq (F_1,F_2)$  and for any $t \in \mathcal{T}$, $M_t \triangleq (M_t^{(1)},M_t^{(2)})$. By~\cite[Th. 3]{yoshida2018optimal}, this defines a coding scheme with access function $g$ such that for any $t\in \mathcal{T}$, $\frac{\lambda_t (\alpha,z,\tau)}{H(F)} = \Delta_{g_1} + \Delta_{g_2},$
where for $i \in \{ 1,2\}$, $\Delta_{g_i} \triangleq \max_{t \in \llbracket  0, T-1 \rrbracket} \left( g_i(t+1) - g_i(t) \right)$. Hence, by remarking that $\Delta_{g_1} = \frac{\alpha}{z}$ and $\Delta_{g_2} = \frac{1-\alpha}{\tau - z} - \frac{\alpha}{z}$, we obtain for any $t\in \mathcal{T}$,
$\frac{\lambda_t (\alpha,z,\tau)}{H(F)} = \frac{1-\alpha}{\tau - z}. $ 
Moreover, $\frac{\rho (\alpha,z,\tau)}{H(F)} 
   = \frac{\rho^{(1)} + \rho^{(2)}}{H(F)}
  = \frac{H(F_2) }{H(F)} \frac{z}{\tau -z} 
 = \frac{n_2 }{n_1+n_2} \frac{z}{\tau -z} 
   =  \left(1-\frac{\alpha}{z} \tau  \right)  \frac{z}{\tau - z} 
 = \frac{z -  \tau \alpha}{\tau - z}. $

\section{Converse proof of Theorem \ref{thconverse}} \label{secproofc1}
Under the leakage symmetry Condition \eqref{eqlt}, we prove lower bounds on the individual share size and the necessary amount of local randomness at the encoder in Sections \ref{secesss} and \ref{secproofc2}, respectively. 
\subsection{Lower bound on individual share size} \label{secesss}
Let $\tau \in \mathcal{T}$, $\alpha \in [0,1]$, $z \in \llbracket 1 , \tau -1 \rrbracket$, and consider an $(\alpha,z)$-private $(\tau,(\lambda_t)_{t\in\mathcal{T}},\rho)$ coding scheme for some $(\lambda_t)_{t\in \mathcal{T}}\in \mathbb{N}^T$, $\rho \in \mathbb{N}$, as defined in Definition \ref{def} under the leakage symmetry Condition \eqref{eqlt}.  In Sections \ref{secsubb} and \ref{secsuba}, we prove that for any $t\in \mathcal{T}$,
\begin{align} 
\frac{\lambda_t}{H(F)} &\geq \frac{1-\alpha}{\tau-z}, \label{eqlambda2}\\
\text{ and }\frac{\lambda_t}{H(F)} &\geq \frac{1}{\tau}, \label{eqlambda1}
\end{align}
 respectively. We will thus deduce from \eqref{eqlambda2} and \eqref{eqlambda1} that for any $t\in \mathcal{T}$, $\frac{\lambda_t}{H(F)} \geq \max \left( \frac{1-\alpha}{\tau-z}, \frac{1}{\tau} \right)$.

\subsubsection{Proof of the first lower bound \eqref{eqlambda2} on $\lambda_t$}\label{secsubb}
Fix $t \in \mathcal{T}$. For $i \in \llbracket z , \tau -1\rrbracket$, define $\mathcal{S}_i \triangleq \begin{cases} \llbracket 1,i\rrbracket & \text{ if } t>i \\ \llbracket 1,i +1 \rrbracket \backslash \{ t \} & \text{ if } t \leq i \end{cases}$  and $\mathcal{S}_{\tau} \triangleq \mathcal{S}_{\tau-1} \cup \{ t\}$. Then, for $i \in \llbracket  z+1, \tau -1 \rrbracket$, we have
\begin{align}
	&H(M_t|M_{\mathcal{S}_i}) \nonumber \\
	&\stackrel{(a)}= H(M_t F|M_{\mathcal{S}_i}) - H(F|M_{\mathcal{S}_i}M_t)  \nonumber \\
	&\stackrel{(b)}= H(F|M_{\mathcal{S}_i}) + H(M_t |F M_{\mathcal{S}_i}) - H(F|M_{\mathcal{S}_i}M_t)  \label{eq1} \\ \nonumber
	&\stackrel{(c)}= (1- C_{i}) H(F)+ H(M_t |F M_{\mathcal{S}_i})  - (1- C_{i+1}) H(F)  \\ 
	&= (C_{i+1}- C_{i}) H(F)  + H(M_t |F M_{\mathcal{S}_i}), \label{eq2}
\end{align}
where\begin{enumerate}[(a)]
\item and (b) hold by the chain rule; 
  \setcounter{enumi}{2}
\item  holds for some constants $C_{i}$ and $C_{i+1}$ by \eqref{eqlt}.
\end{enumerate} Next, we~have
\begin{align}	&H(M_t|M_{\mathcal{S}_{\tau}}) \nonumber \\
	&\stackrel{(a)}= H(F|M_{\mathcal{S}_{\tau}})  + H(M_t |F M_{\mathcal{S}_{\tau}}) - H(F|M_{\mathcal{S}_{\tau}}M_t)  \nonumber \\ 
	&\stackrel{(b)}=   H(M_t |F M_{\mathcal{S}_{\tau}})   \nonumber \\ 
		&\stackrel{(c)}= (C_{\tau+1}- C_{\tau}) H(F)  + H(M_t |F M_{\mathcal{S}_{\tau}}), 
 \label{eq4}
\end{align}
where
\begin{enumerate}[(a)]
\item holds as in \eqref{eq1}; 
\item holds by \eqref{eqreq1};
\item  holds by defining $C_{\tau+1} \triangleq C_{\tau} = 1$.
\end{enumerate}
 We also have
\begin{align}
	&H(M_t|M_{\mathcal{S}_z}) \nonumber\\
	&\stackrel{(a)}= H(F|M_{\mathcal{S}_z})  + H(M_t |F M_{\mathcal{S}_z}) - H(F|M_{\mathcal{S}_z}M_t)  \nonumber \\ \nonumber
	&\stackrel{(b)}\geq  (1- \alpha) H(F)  + H(M_t |F M_{\mathcal{S}_z}) - (1- C_{z+1}) H(F)  \\
	&= (C_{z+1}- \alpha) H(F)  + H(M_t |F M_{\mathcal{S}_z}) , \label{eq3}
\end{align}
where 
\begin{enumerate}[(a)]
\item holds as in \eqref{eq1}; 
\item  holds by \eqref{eqreq2} and \eqref{eqlt}. 
\end{enumerate}
In the following, for convenience, we define  $C_z \triangleq \alpha$.  
Next, we have
\begin{align}
	&H(M_t) \nonumber \\
	& \stackrel{(a)}\geq H(M_t|M_{\mathcal{S}_{z}})  \label{eqinter} \\ \nonumber
	& \stackrel{(b)}= H(M_t|M_{\mathcal{S}_{z}}) - H(M_t | M_{\mathcal{S}_{\tau}})\\ \nonumber
	& = \sum_{i=z}^{\tau-1} \left(H(M_t|M_{\mathcal{S}_i}) - H(M_t|M_{\mathcal{S}_{i+1}}) \right)\\ \nonumber 
	& \stackrel{(c)}\geq \sum_{i=z}^{\tau-1} \left[ (C_{i+1}- C_{i}) H(F)  + H(M_t |F M_{\mathcal{S}_i}) \right. \\ \nonumber
 & \phantom{--} \left. - (C_{i+2}- C_{i+1}) H(F)  - H(M_t |F M_{\mathcal{S}_{i+1}}) \right]^+\\ \nonumber
	& \stackrel{(d)} \geq H(F) \sum_{i=z}^{\tau-1}  [ (C_{i+1}- C_{i})  - (C_{i+2} - C_{i+1}) ]^+\\ \nonumber
	& \stackrel{(e)}= H(F) \sum_{i=z+1}^{\tau}  [ (\phi(i)- \phi(i-1))  - (\phi(i+1) - \phi(i)) ]^+ \\
    & \stackrel{(f)}\geq  H(F) \min_{\phi \in \mathcal{F}} \sum_{i=z+1}^{\tau}  [ (\phi(i)- \phi(i-1))  - (\phi(i+1) - \phi(i)) ],\!\!\!^+ \label{eqmin}
	\end{align}
	where \begin{enumerate}[(a)]
\item holds because conditioning does not increase entropy;
\item  holds because $t \in \mathcal{S}_{\tau}$; 
\item  holds because $H(M_t|M_{\mathcal{S}_i}) - H(M_t|M_{\mathcal{S}_{i+1}}) \geq 0$ (conditioning does not increase entropy and $\mathcal{S}_i \subset \mathcal{S}_{i+1}$) and by \eqref{eq2}, \eqref{eq4}, \eqref{eq3};
\item holds because conditioning does not increase entropy; 
\item holds with the function $\phi: \llbracket z,\tau+1 \rrbracket \to [0,1]$ defined such that $\phi(i) = C_i$ for $i \in \llbracket z, \tau+1 \rrbracket$;
\item   holds with the minimum taken over the set $\mathcal{F}$ of  all the functions $\phi: \llbracket z,\tau+1 \rrbracket \to [0,1]$ that are non-decreasing (by \eqref{eqlt} because for any $\mathcal{S} \subset \mathcal{S}' \subset \mathcal{T}$, $\frac{I(F;M_{\mathcal{S}})}{H(F)} \leq \frac{I(F;M_{\mathcal{S}'})}{H(F)}$) and such that $\phi(z) = \alpha$ (because $ C_z = \alpha$), $\phi(\tau+1)=\phi(\tau)=1$ (because $C_{\tau+1} = C_{\tau} = 1$).  
	\end{enumerate}
We now lower bound the minimum in the right-hand side of \eqref{eqmin} by an expression that only depends on the concave envelopes of the access functions that appear in the objective function. This  allows us to conclude that a piecewise linear access function is solution to the optimization. Specifically, let $\phi \in \mathcal{F}$ and let $\phi^+$ be the concave envelope of $\phi$ over $\llbracket z,\tau+1 \rrbracket$, i.e., for $i \in \llbracket z,\tau+1 \rrbracket$, $\phi^+(i) \triangleq \min \{ \psi(i) : \psi \geq \phi, \psi \text{ is concave}\}$. Note that $\phi^+(z) = \phi(z)$ and $\phi^+(\tau+1) = \phi(\tau+1)$. Then, for any $i \in \llbracket z+1,\tau \rrbracket$ such that $\phi(i) = \phi^+(i)$, we~have 
	\begin{align}
		&[ (\phi(i)- \phi(i-1))  - (\phi(i+1) - \phi(i)) ]^+ \nonumber \\
		& \geq  (\phi(i)-\phi(i-1))  - (\phi(i+1) - \phi(i))  \nonumber \\ \nonumber
		&\stackrel{(a)} \geq (\phi(i)- \phi^+(i-1))  - (\phi^+(i+1) - \phi(i)) \\
		& \stackrel{(b)}= (\phi^+(i)- \phi^+(i-1))  - (\phi^+(i+1) - \phi^+(i)) , \label{eqm1}
	\end{align}
		where \begin{enumerate}[(a)]
\item holds because $\phi^+ \geq \phi$;
\item holds because $\phi(i) = \phi^+(i)$.
\end{enumerate}
 Moreover, for any $i \in \llbracket z+1,\tau \rrbracket$ such that $\phi(i) \neq \phi^+(i)$, we have 
	\begin{align}
	    & [ (\phi(i)- \phi(i-1))  - (\phi(i+1) - \phi(i))  ]^+  \nonumber \\ \nonumber
		& \geq  0 \\
		& = (\phi^+(i)- \phi^+(i-1))  - (\phi^+(i+1) - \phi^+(i)), \label{eqm2}
	\end{align}
	where the last equality holds because $\phi^+$ is linear between $i-1$ and $i+1$, i.e., $\phi^+(i)- \phi^+(i-1)  = \phi^+(i+1) - \phi^+(i)$ - details are provided in  Appendix \ref{AppendixA}. 
	
	Next, we have 
	\begin{align}
		& \sum_{i=z+1}^{\tau}  [ (\phi(i)- \phi(i-1))  - (\phi(i+1) - \phi(i)) ]^+ \nonumber \\ \nonumber
		& \stackrel{(a)}  \geq 		 \sum_{i=z+1}^{\tau}  (\phi^+(i)- \phi^+(i-1))  - (\phi^+(i+1) - \phi^+(i)) )\\\nonumber
 		& = \phi^+(z+1)- \phi^+(z) - \phi^+(\tau+1) + \phi^+(\tau) \\ \nonumber
 		 		& \stackrel{(b)}= \phi^+(z+1)- \phi^+(z) \\ 
 		 		 		 		& \stackrel{(c)} \geq  \frac{1-\alpha}{\tau-z}  , \label{eqf1}
	\end{align}
	where
	\begin{enumerate}[(a)]
\item holds by \eqref{eqm1} and \eqref{eqm2};
\item holds because $\phi^+(\tau+1)= \phi^+(\tau) = 1$;
\item holds  because 
	$\phi^+(z+1) - \phi^+(z) \geq (\phi^+(\tau) - \phi^+(z))/(\tau-z)$ by concavity of $\phi^+$ and where we have used that $\phi^+(\tau) =1$ and $\phi^+(z)=\phi(z)= \alpha$.
	\end{enumerate}
		Finally, we have
	\begin{align*}
\lambda_t &\geq H(M_t) \\
& \geq H(F) \frac{1-\alpha}{\tau-z},
	\end{align*}
	where the last inequality holds by \eqref{eqmin} and \eqref{eqf1}, which is valid for any $\phi \in \mathcal{F}$.

\subsubsection{Proof of the second lower bound \eqref{eqlambda1} on $\lambda_t$} \label{secsuba}
Note that in the proof of  \eqref{eqlambda2}, one can substitute the variable $z$ by zero such that one can show 
\begin{align*}
\lambda_t &\geq H(M_t) \\
& \geq H(F)( \phi^+(1)- \phi^+(0))\\
& \geq H(F) \frac{1}{\tau},
	\end{align*}
	where the last inequality holds because 
	$\phi^+(1) - \phi^+(0) \geq (\phi^+(\tau) - \phi^+(0))/\tau$ by concavity of $\phi^+$ and where we have used that $\phi^+(\tau) =1$ and $\phi^+(0)=0$.
	
\subsection{Lower bound on the amount of local randomness} \label{secproofc2}

Let $\tau \in \mathcal{T}$, $\alpha \in [0,1]$, $z \in \llbracket 1 , \tau -1 \rrbracket$, and consider an $(\alpha,z)$-private $(\tau,(\lambda_t)_{t\in\mathcal{T}},\rho)$ coding scheme for some $(\lambda_t)_{t\in \mathcal{T}}\in \mathbb{N}^T$, $\rho \in \mathbb{N}$, as defined in Definition \ref{def} under the leakage symmetry Condition \eqref{eqlt}. Then, we have
\begin{align*}
& \rho + H(F)\\
& \stackrel{(a)}{=} H(R) + H(F)\\ 
& \stackrel{(b)}{=}  H(RF) \\
&\stackrel{(c)}{\geq}  H(M_{\mathcal{T}}) \\
&\stackrel{(d)} = H(M_{\llbracket 1,z \rrbracket}) + H(M_{\mathcal{T} \backslash \llbracket 1,z\rrbracket } |M_{\llbracket 1,z \rrbracket} )\\
& \stackrel{(e)}{=} \sum_{t=1}^z H(M_t | M_{\llbracket 1,t-1\rrbracket})+ H(M_{\mathcal{T} \backslash \llbracket 1,z\rrbracket } |M_{\llbracket 1,z \rrbracket} )  \\
& \stackrel{(f)}{\geq} \sum_{t=1}^z H(M_t | M_{\mathcal{S}_{z,t} }) + H(M_{\mathcal{T} \backslash \llbracket 1,z\rrbracket} |M_{\llbracket 1,z \rrbracket} ) \\
& \stackrel{(g)}{\geq} z\frac{1-\alpha}{\tau-z} H(F) + H(M_{\mathcal{T} \backslash \llbracket 1,z \rrbracket} |M_{\llbracket 1,z \rrbracket} ) \\
& \stackrel{(h)}{=} z\frac{1-\alpha}{\tau-z} H(F) + H(M_{\mathcal{T} \backslash \llbracket 1,z \rrbracket } F|M_{\llbracket 1,z \rrbracket} ) \\
&\phantom{--}-H(F|M_{\mathcal{T}\backslash \llbracket 1,z \rrbracket } M_{\llbracket 1,z \rrbracket} )\\
& \stackrel{(i)}{\geq} z\frac{1-\alpha}{\tau-z} H(F) + H( F|M_{\llbracket 1,z \rrbracket} ) \\
& \stackrel{(j)}\geq z\frac{1-\alpha}{\tau-z}H(F) + (1-\alpha)H(F) \\
& = \tau \frac{1- \alpha}{\tau-z}H(F), \numberthis \label{eqrhofinal}
\end{align*}
where 
\begin{enumerate}[(a)]
\item holds by uniformity of $R$;
\item holds by independence between $F$ and $R$;
\item holds because $M_{\mathcal{T}}$ is a deterministic function of $(R,F)$;
\item and (e) hold by the chain rule;   \setcounter{enumi}{5}
\item holds because conditioning does not increase entropy and we have defined $\mathcal{S}_{z,t} \triangleq \llbracket 1,z+1 \rrbracket \backslash \{ t\}$ for $t \in \llbracket 1 ,z \rrbracket $;
\item holds because for any $t \in \llbracket 1 ,z \rrbracket $, $H(M_t | M_{\mathcal{S}_{z,t} }) \geq \frac{1-\alpha}{\tau-z}H(F)$, by the converse proof of Theorem \ref{thconverse} starting from \eqref{eqinter};
\item holds by the chain rule;
\item holds because $H(F|M_{\mathcal{T}\backslash \llbracket 1,z \rrbracket } M_{\llbracket 1,z \rrbracket} ) = H(F|M_{\mathcal{T}})= 0$ by \eqref{eqreq1}, and $H(M_{\mathcal{T} \backslash \llbracket 1,z \rrbracket } F|M_{\llbracket 1,z \rrbracket} ) \geq H(F|M_{\llbracket 1,z \rrbracket} ) $ by the chain rule and positivity of conditional entropy;
\item holds by \eqref{eqreq2}.
\end{enumerate}
 Finally, from \eqref{eqrhofinal}, we have
\begin{align*}
 \rho 
\geq \left(\tau\frac{1-\alpha}{\tau-z} -1\right)H(F)   =  \frac{z- \tau \alpha }{\tau-z}H(F),
\end{align*}
and since we also have $\rho \geq 0$, we conclude
\begin{align*}
 \frac{\rho}{H(F)} 
&\geq  \frac{[z- \tau \alpha]^+ }{\tau-z}.
\end{align*}

\section{Concluding remarks} \label{concl}
We considered a setting where a file must be stored in $L$ servers such that: (i) any $\tau$ servers that pool their information together can reconstruct the file, and (ii) any $z$ servers cannot learn more than a fraction $\alpha\in[0,1]$ of the file, where $\tau$, $z$, and $\alpha$ are parameters to be chosen by the system designer. This setting generalizes ramp secret sharing in that  information leakage about the file  is allowed up to a fraction $\alpha$, and goes beyond existing works on uniform secret sharing by considering share size optimization over a set of access functions  rather than for a fixed access function. Specifically, for given parameters $\tau$, $z$, $\alpha$, and under the leakage symmetry assumption that any set of colluding servers must have the same information leakage about the file that any other set of colluding servers of same size, we derived the optimal individual share size at each server. In the absence of any leakage symmetry, we also derived the optimal sum of the share sizes at all the servers and the optimal amount of local randomness needed at the encoder. As a byproduct, in the case $\alpha=0$, our results prove that among all uniform secret sharing schemes for our model, linear ramp secret sharing schemes require the smallest individual share size.

\appendices

\section{Proof of \eqref{eqm2}} \label{AppendixA}

By contradiction, assume that  $\phi^+$ is not linear between $i-1$ and $i+1$, then we must have  
	\begin{align}
	\phi^+(i) > \frac{\phi^+(i+1)+ \phi^+(i-1)}{2} \label{eqcontrad}
\end{align}	
	 since $\phi^+$ is concave. Next, we have a contradiction by constructing $\psi_i$, a concave function such that $\phi \leq \psi_i < \phi^+$, as follows: 
	\begin{align*}
	\psi_i : \llbracket z , \tau +1 \rrbracket &\to \mathbb{R}\\
	 j & \mapsto \begin{cases}
	  \phi^+(j) & \text{if } j \neq i\\
	  \max \left( \frac{\phi^+(i+1)+ \phi^+(i-1)}{2} , \phi(i) \right) &\text{if } j =i
	 \end{cases}.
	\end{align*}
We have $\phi \leq \psi_i$ (since $\phi \leq \phi^+$), and 	$\psi_i < \phi^+$ by \eqref{eqcontrad} and because $\phi^+(i) > \phi(i)$ (since $\phi^+ \geq \phi$ and $\phi^+(i) \neq \phi(i)$). Then, to show concavity of $\psi_i$, it is sufficient to show that  $\psi_i^{\Delta}$ is non-increasing, where $\psi_i^{\Delta}$ is defined as 
\begin{align*}
\psi_i^{\Delta}:\llbracket z, \tau  \rrbracket &\to \mathbb{R}\\
	 j & \mapsto \psi_i(j+1) - \psi_i(j).
\end{align*}
For $j \in \llbracket z , i -3 \rrbracket \cup \llbracket i+1, \tau \rrbracket $, we have 
\begin{align} \label{eqmonotone1}
\psi_i^{\Delta}(j+1) \leq \psi_i^{\Delta}(j) 
\end{align} by definition of  $\psi_i^{\Delta}$ and concavity of $\phi^+$. Then, we have
\begin{align*}
\psi_i^{\Delta}(i-1) 
& \stackrel{(a)} = \psi_i(i) - \psi_i(i-1)  \\
&\stackrel{(b)} = \psi_i(i) -\phi^+(i-1) \\
& \stackrel{(c)}\leq \phi^+(i) -\phi^+(i-1) \\
& \stackrel{(d)}\leq \phi^+(i-1) - \phi^+(i-2) \\
& \stackrel{(e)}= \psi_i(i-1) - \psi_i(i-2) \\
& \stackrel{(f)}= \psi_i^{\Delta}(i-2),
\end{align*}
where \begin{enumerate}[(a)]
\item and (f) hold by definition of $\psi_i^{\Delta}$;
\item and (e) hold by definition of $\psi_i$;
\item holds because $\psi_i < \phi^+$;
\item holds by concavity of $\phi^+$.\end{enumerate}
 Then, we have
\begin{align*}
\psi_i^{\Delta}(i) 
&\stackrel{(a)}= \psi_i(i+1) - \psi_i(i)  \\
&\stackrel{(b)}= \phi^+(i+1) - \psi_i(i)  \\
& \stackrel{(c)} \leq \psi_i(i) - \phi^+(i-1) \\
&\stackrel{(d)} = \psi_i(i) - \psi_i(i-1) \\
&\stackrel{(e)} = \psi_i^{\Delta}(i-1), \numberthis \label{eqmonotone2}
\end{align*}
where 
\begin{enumerate}[(a)]
\item and (e) hold by definition of $\psi_i^{\Delta}$;
\item and (d) hold by definition of $\psi_i$;
\item holds because $\frac{\phi^+(i+1)+ \phi^+(i-1)}{2} \leq \psi_i (i)$.
\end{enumerate}
 Then, we also have
\begin{align*}
\psi_i^{\Delta}(i+1) 
& \stackrel{(a)}= \psi_i(i+2) - \psi_i(i+1)  \\
&\stackrel{(b)}= \phi^+(i+2) - \phi^+(i+1)  \\
&\stackrel{(c)} \leq \phi^+(i+1) - \phi^+(i)  \\
&\stackrel{(d)} \leq \phi^+(i+1) - \psi_i(i) \\
&\stackrel{(e)} = \psi_i(i+1) - \psi_i(i) \\
&\stackrel{(f)} = \psi_i^{\Delta}(i), \numberthis \label{eqmonotone3}
\end{align*}
where 
\begin{enumerate}[(a)]
\item and (f) hold by definition of $\psi_i^{\Delta}$;
\item and (e) hold by definition of $\psi_i$; 
\item holds by concavity of $\phi^+$; 
\item holds because $\psi_i < \phi^+$.
\end{enumerate}
 Hence, by \eqref{eqmonotone1}, \eqref{eqmonotone2}, and \eqref{eqmonotone3}, $\psi_i^{\Delta}$ is non-increasing and we have thus proved \eqref{eqm2} by contradiction.

 \section{Converse proof of Theorem \ref{th3}} \label{secconvth3}

Let $\tau \in \mathcal{T}$, $\alpha \in [0,1]$, $z \in \llbracket 1 , \tau -1 \rrbracket$, and consider an $(\alpha,z)$-private $(\tau,(\lambda_t)_{t\in\mathcal{T}},\rho)$ coding scheme for some $(\lambda_t)_{t\in \mathcal{T}}\in \mathbb{N}^T$, $\rho \in \mathbb{N}$, as defined in Definition \ref{def}. We prove the lower bounds
\begin{align}
\frac{\sum_{t\in \mathcal{T}}\lambda_t}{H(F)} &\geq T \max \left( \frac{1-\alpha}{\tau-z}, \frac{1}{\tau} \right), \label{eqconva1}\\
\frac{\rho}{H(F)} &\geq \frac{[z - \tau \alpha]^+}{\tau-z} , \label{eqconva2}\end{align}
in Appendices \ref{conv1} and \ref{conv2}, respectively. 

\subsection{Lower bound on the sum of the share sizes} \label{conv1}

 For $\mathcal{W} \subseteq \mathcal{T}$ and $\mathcal{S} \subseteq \mathcal{T} \backslash \mathcal{W} $ such that $|\mathcal{W}|=z$ and $|\mathcal{S}|=\tau-z$, we have 
 \begin{align*}
   \sum_{l \in \mathcal{S}} H(M_{l})  
 & \stackrel{(a)}\geq   H(M_{\mathcal{S}})\\
 & \stackrel{(b)}\geq H(M_{\mathcal{S}}| M_{\mathcal{W}} )  \\
 & \geq    I(M_{\mathcal{S}}; F | M_{\mathcal{W}}) \\
 & =    H(F | M_{\mathcal{W}}) - H(F|M_{\mathcal{W}\cup\mathcal{S}} ) \\
 &\stackrel{(c)}=     H( F | M_{\mathcal{W}})   \\
 & =   H(F) - I( F ; M_{\mathcal{W}})   \\
 & \stackrel{(d)}\geq (1-\alpha) H( F)  , \numberthis \label{eqRanDint}
 \end{align*}
 where 
 \begin{enumerate}[(a)]
\item and (b)  hold by the chain rule and because conditioning does not increase entropy;   \setcounter{enumi}{2}
\item  holds by \eqref{eqreq1} because  $|\mathcal{S} \cup \mathcal{W}|=\tau$;
\item holds by \eqref{eqreq2} because $|\mathcal{W}|=z$. 
\end{enumerate}
Then, by defining $\Theta \triangleq \frac{T}{\tau-z} {T \choose z}^{-1} {{T-z}\choose{\tau-z}}^{-1}$,  we have
   \begin{align*} 
& T\frac{1 - \alpha}{\tau-z} H( F)\\ 
 & \stackrel{(a)}= \Theta \sum_{ \substack{\mathcal{W} \subseteq \mathcal{T}  \\  |\mathcal{W} | = z }} \sum_{ \substack{\mathcal{S} \subseteq  \mathcal{W}^c \\  |\mathcal{S} | = \tau-z }}  (1-\alpha)H( F ) \\ 
   &  \stackrel{(b)} \leq \Theta \sum_{ \substack{\mathcal{W} \subseteq \mathcal{T}  \\  |\mathcal{W} | = z }} \sum_{ \substack{\mathcal{S} \subseteq  \mathcal{W}^c \\  |\mathcal{S} | = \tau-z }}  \sum_{l \in \mathcal{S}} H(M_{l})  \\
   & \stackrel{(c)} = \Theta \sum_{ \substack{\mathcal{W} \subseteq \mathcal{T}  \\  |\mathcal{W} | = z }} { T-z - 1 \choose \tau-z-1 }  \sum_{l \in  \mathcal{W}^c } H(M_{l})  \\
   &  \stackrel{(d)}  = \Theta  { T-z - 1 \choose \tau-z-1 } \sum_{ \substack{\mathcal{W} \subseteq \mathcal{T}  \\  |\mathcal{W} | =T- z }}   \sum_{l \in   \mathcal{W}} H(M_{l})  \\
      & \stackrel{(e)}= \Theta  { T-z - 1 \choose \tau-z-1 } { T - 1 \choose T-z-1 }   \sum_{l \in   \mathcal{T}} H(M_{l})  \\
            & =  \sum_{l \in   \mathcal{T}} H(M_{l})  \\
            & \stackrel{(f)} \leq  \sum_{l \in   \mathcal{T}} \lambda_l,  \numberthis \label{eqratecsum}
 \end{align*}
 where 
\begin{enumerate}[(a)]
\item holds because ${T \choose z}^{-1} {{T-z}\choose{\tau-z}}^{-1} \sum_{ \substack{\mathcal{W} \subseteq \mathcal{T}  \\  |\mathcal{W} | = z }} \sum_{ \substack{\mathcal{S} \subseteq  \mathcal{W}^c \\  |\mathcal{S} | = \tau-z }} 1 = 1$;
 \item holds by \eqref{eqRanDint};
\item holds because for any $l\in\mathcal{W}^c$, $H(M_{l})$ appears exactly ${ T-z - 1 \choose \tau-z-1 }$ times in the term $\sum_{ \substack{\mathcal{S} \subseteq \mathcal{W}^c  \\  |\mathcal{S} | = \tau-z }}  \sum_{l \in \mathcal{S}} H(M_{l})$, note that this observation is similar to \cite[Lemma 3.2]{de1999multiple};
\item holds by a change of variables in the sums; 
\item holds because for any $l\in\mathcal{T}$, $H(M_{l})$ appears exactly ${ T - 1 \choose T-z-1 }$ times in the sum $\sum_{ \substack{\mathcal{W} \subseteq \mathcal{T}  \\  |\mathcal{W} | =T- z }}   \sum_{l \in   \mathcal{W}} H(M_{l})$;
\item holds by definition of $M_{l}$, $l\in\mathcal{T}$. 
 \end{enumerate}
Then,
 for any $\mathcal{S} \subseteq \mathcal{T}$ such that $|\mathcal{S}|=\tau$, we have
\begin{align}
 \sum_{l \in\mathcal{S}} H(M_{l}) 
&\geq H(M_{\mathcal{S}}) \nonumber\\
& \geq H(F), \label{eqfirstlb}
\end{align}
where the last inequality holds because if $H(M_{\mathcal{S}})< H(F)$, then by \eqref{eqreq1}, we have $H(F|M_{\mathcal{S}})=0$ and then the contradiction $H(M_{\mathcal{S}}|F) =H(M_{\mathcal{S}}) -H(F)<0$. We also have
\begin{align*}
\frac{T}{\tau} H( F) 
 & = \frac{T}{\tau}  {{T}\choose{\tau}}^{-1}  \sum_{ \substack{\mathcal{S} \subseteq  \mathcal{T} \\  |\mathcal{S} | = \tau }}  H( F ) \\ 
   &  \stackrel{(a)} \leq  \frac{T}{\tau}  {{T}\choose{\tau}}^{-1}  \sum_{ \substack{\mathcal{S} \subseteq  \mathcal{T} \\  |\mathcal{S} | = \tau }}  \sum_{l\in \mathcal{S}} H(M_{l})  \\
      &  =  \frac{T}{\tau}  {{T}\choose{\tau}}^{-1}    {{T-1}\choose{\tau-1}} \sum_{l\in \mathcal{T}} H(M_{l})  \\
            & =  \sum_{l \in   \mathcal{T}} H(M_{l})  \\
            & \stackrel{(b)} \leq  \sum_{l \in   \mathcal{T}} \lambda_l,  \numberthis \label{eqratecsum2}
            \end{align*}
where\begin{enumerate}[(a)]
 \item holds by \eqref{eqfirstlb};
\item holds by definition of $M_{l}$, $l\in\mathcal{T}$. 
 \end{enumerate}

Finally, we conclude from \eqref{eqratecsum} and \eqref{eqratecsum2} that \eqref{eqconva1}~holds. 
 
\subsection{Lower bound on the amount of local randomness} \label{conv2}
Let $\mathcal{V} \subseteq \mathcal{T}$ such that $v \triangleq |\mathcal{V}|< z$. For $\mathcal{W} \subseteq \mathcal{T} \backslash \mathcal{V}$ and $\mathcal{S} \subseteq \mathcal{T} \backslash (\mathcal{W} \cup \mathcal{V}) $ such that $|\mathcal{W}|= z - v$ and $|\mathcal{S}|=\tau-z$, we have 
 \begin{align*}
   \sum_{l \in \mathcal{S}} H(M_{l}  |M_{\mathcal{V}} )
 &\stackrel{(a)}  \geq   H(M_{\mathcal{S}} |M_{\mathcal{V}} )\\
 & \stackrel{(b)}\geq H(M_{\mathcal{S}} |M_{\mathcal{V}\cup \mathcal{W}} )   \\
    & \stackrel{(c)} \geq    (1- \alpha) H( F) , \numberthis \label{eqRanDint0}
 \end{align*}
where 
\begin{enumerate}[(a)]
\item and (b) hold by the chain rule and because conditioning does not increase entropy;  \setcounter{enumi}{2}
\item  holds similar to \eqref{eqRanDint} with the substitution $\mathcal{W} \leftarrow \mathcal{V} \cup \mathcal{W}$, which is possible because $|\mathcal{V}\cup \mathcal{W}\cup \mathcal{S}| = \tau$ and $|\mathcal{V}\cup \mathcal{W}|=z$.
\end{enumerate}
Then, by defining $\Lambda \triangleq \frac{1}{\tau-z} {T-v \choose z-v}^{-1} {{T-z}\choose{\tau-z}}^{-1} $, we have
   \begin{align*} 
& \frac{1-\alpha}{\tau-z} H( F )\\
 & = \Lambda \sum_{ \substack{\mathcal{W} \subseteq \mathcal{T}\backslash \mathcal{V}  \\  |\mathcal{W} | = z -v}} \sum_{ \substack{\mathcal{S} \subseteq  \mathcal{T} \backslash (\mathcal{W} \cup \mathcal{V}) \\  |\mathcal{S} | = \tau-z }} (1-\alpha)  H( F ) \\ 
   &  \stackrel{(a)} \leq \Lambda \sum_{ \substack{\mathcal{W} \subseteq \mathcal{T}\backslash \mathcal{V}  \\  |\mathcal{W} | = z -v}} \sum_{ \substack{\mathcal{S} \subseteq  \mathcal{T} \backslash (\mathcal{W} \cup \mathcal{V}) \\  |\mathcal{S} | = \tau-z }} \sum_{l \in \mathcal{S}}  H(M_{l}  |M_{\mathcal{V}} ) \\
      &  \stackrel{(b)} = \Lambda\sum_{ \substack{\mathcal{W} \subseteq \mathcal{T}\backslash \mathcal{V}  \\  |\mathcal{W} | = z -v}}   { T - z  -  1 \choose \tau - z - 1 }   \sum_{ l \in  \mathcal{T} \backslash (\mathcal{W} \cup \mathcal{V})  }     H(M_{l}  |M_{\mathcal{V}} ) \\
            &  \stackrel{(c)} = \Lambda{ T-z - 1 \choose \tau-z-1 } \sum_{ \substack{\mathcal{W} \subseteq \mathcal{T} \backslash  \mathcal{V}  \\  |\mathcal{W} | = T -z}} \sum_{ l \in  \mathcal{W}   } H(M_{l}  |M_{\mathcal{V}} ) \\
                        &  \stackrel{(d)} = \Lambda{ T-z  - 1 \choose \tau-z-1 } {T-v -1\choose T-z-1}   \sum_{ l \in  \mathcal{T} \backslash   \mathcal{V}  }  H(M_{l}  |M_{\mathcal{V}} )  \\
      & = \frac{1}{T-v}\sum_{ l \in  \mathcal{T} \backslash   \mathcal{V}  } H(M_{l}  |M_{\mathcal{V}} ) \\
    & \stackrel{(e)} \leq  \frac{1}{T-v}\sum_{ l \in  \mathcal{T} \backslash   \mathcal{V}  }  H(M_{l^\star(\mathcal{V})}  |M_{\mathcal{V}} ) \\
        & =    H(M_{l^\star(\mathcal{V})}  |M_{\mathcal{V}} ) \\
        & =   H(M_{\mathcal{T}}  |M_{\mathcal{V}} ) -  H(M_{\mathcal{T}}  |M_{\mathcal{V} \cup \{l^\star(\mathcal{V})\}}  )  , \numberthis \label{eqint0}
 \end{align*}
 where 
 \begin{enumerate}[(a)]
\item holds by \eqref{eqRanDint0}; 
\item holds because for any $l\in\mathcal{T} \backslash (\mathcal{W} \cup \mathcal{V})$, the term $ H(M_{l}|M_{\mathcal{V}} )$ appears exactly ${ T-z - 1 \choose \tau-z-1 }$ times in the term $\sum_{ \substack{\mathcal{S} \subseteq  \mathcal{T} \backslash (\mathcal{W} \cup \mathcal{V}) \\  |\mathcal{S} | = \tau-z }} \sum_{l \in \mathcal{S}}  H(M_{l} |M_{\mathcal{V}} )$, this argument is similar to \cite[Lemma 3.2]{de1999multiple};
\item holds by a change of variables in the sums;
\item holds because for any $l\in\mathcal{T}\backslash \mathcal{V} $, $ H(M_{l} |M_{\mathcal{V}} ) $ appears exactly ${ T - v-1 \choose T-z-1 }$ times in the term $\sum_{ \substack{\mathcal{W} \subseteq \mathcal{T} \backslash \mathcal{V}  \\  |\mathcal{W} | =T- z }}   \sum_{l \in   \mathcal{W}}  H(M_{l} | M_{\mathcal{V}} )$;
\item holds with $l^\star(\mathcal{V})  \in \argmax_{l\in\mathcal{T} \backslash \mathcal{V}}  H(M_{l} |M_{\mathcal{V}})$.
\end{enumerate}
Next, define $\mathcal{V}_0 \triangleq \emptyset$ and for $i\in \llbracket 1, z \rrbracket$, $\mathcal{V}_i \triangleq \mathcal{V}_{i-1} \cup \{ l^{\star} (\mathcal{V}_{i-1}) \}$. Then, we~have
 \begin{align*} 
   &\frac{z - \tau \alpha}{\tau-z} H(F) \\
      & =  \left(z\frac{1 - \alpha}{\tau-z} - \alpha \right) H(F)  \\
  &  =  - \alpha H(F) + z\frac{1 -  \alpha}{\tau-z} H( F )\\
 & = - \alpha H(F) + \sum_{i=0}^{z-1} \frac{1-  \alpha}{\tau-z} H( F )\\
        &\stackrel{(a)} \leq  - \alpha H(F) + \sum_{i=0}^{z-1}[ H(M_{\mathcal{T}}  | M_{\mathcal{V}_i}  )  - H(M_{\mathcal{T}}  | M_{\mathcal{V}_{i+1}} ) ]   \\
                & = - \alpha H(F) +  H(M_{\mathcal{T}}  )  - H(M_{\mathcal{T}}  | M_{\mathcal{V}_{z}} ) \\
                                    & \stackrel{(b)} \leq  - \alpha H(F) +  H( F,R  )  - H(M_{\mathcal{T}}  | M_{\mathcal{V}_{z}} ) \\
                    & \stackrel{(c)} =  (1- \alpha) H(F) + H(R) - H(M_{\mathcal{T}}  | M_{\mathcal{V}_{z}} )  \\
                                        & \stackrel{(d)} =  (1- \alpha) H(F) + H(R) - H(FM_{\mathcal{T}}  | M_{\mathcal{V}_{z}} )  \\                                                                            
                                        &  \leq  (1- \alpha) H(F) + H(R) - H(F | M_{\mathcal{V}_{z}} )  \\                                                                            
 &\stackrel{(e)} \leq   H(R) \\
 & = \rho  , \numberthis \label{eqratelrand}
 \end{align*}
 where 
 \begin{enumerate}[(a)]
\item holds by applying $z$ times \eqref{eqint0}  and the definition of $\mathcal{V}_i$, $i \in \llbracket 0, z \rrbracket$;
\item holds because $M_{\mathcal{T}}$ is a deterministic function of $(F,R)$
\item holds by independence between $F$ and $R$; 
\item  holds by the chain rule and because $H(F  | M_{\mathcal{T}} ) = 0$ by \eqref{eqreq1}; 
\item holds because $ - H(F  | M_{\mathcal{V}_{z}}  ) \leq - (1-\alpha) H(F)$ by \eqref{eqreq2}.
\end{enumerate}
 Finally, since we also have $\rho \geq 0$, we conclude from \eqref{eqratelrand} that \eqref{eqconva2} holds.

\bibliographystyle{IEEEtran}
\bibliography{bib}

\end{document}